\documentclass[10pt]{article}

\usepackage{chngcntr}
\usepackage{graphicx}
\usepackage{graphics}
\usepackage{amsthm}
\usepackage{amscd}
\usepackage{amsmath}
\usepackage{amsfonts}
\usepackage{amssymb}
\usepackage{dblfloatfix}
\usepackage{caption}
\usepackage[english]{babel}
\usepackage{subcaption}

\theoremstyle{definition}

\newtheorem{theorem}{Theorem}[section]

\newtheorem{corollary}[theorem]{Corollary}

\newtheorem{definition}[theorem]{Definition}

\newtheorem{lemma}[theorem]{Lemma}

\newtheorem{example}[theorem]{Example}

\newtheorem{proposition}[theorem]{Proposition}
\newtheorem{remark}[theorem]{Remark}

\makeatletter
\newcommand*{\rom}[1]{\expandafter\@slowromancap\romannumeral #1@}
\makeatother
\newcommand{\R}{{\cal R}}
\newcommand{\F}{\mathbb{F}}
\newcommand{\Tr}{Tr_{r/q}}

\begin{document}

\title{Two-Weight and a Few Weights  Trace Codes over $\F_{q}+u\F_{q}$ $\footnote{
 E-Mail addresses: hwliu@mail.ccnu.edu.cn (H. Liu), m.youcef@mails.ccnu.edu.cn (Y. Maouche).}$}

\author{Hongwei Liu$^1$,~Youcef Maouche$^{1,2}$}

\date{\small
${}^1$School of Mathematics and Statistics, Central China Normal University,Wuhan, Hubei, 430079, China\\
${}^2$Department of Mathematics, University of Sciences and Technology HOUARI BOUMEDIENE, Alger, Algeria\\
}
\maketitle

\leftskip 0.8in
\rightskip 0.8in
\noindent {\bf Abstract.} Let $p$ be a prime number, $q=p^s$ for a positive integer $s$. For any positive divisor $e$ of $q-1$, we construct an infinite family codes of size $q^{2m}$ with few Lee-weight. These codes are defined as trace codes  over the ring $R=\F_q + u\F_q$, $u^2 = 0$. Using Gauss sums, their Lee weight distributions are provided. When $\gcd(e,m)=1$, we obtain an infinite family of two-weight codes over the finite field $\F_q$ which meet the Griesmer bound. Moreover, when  $\gcd(e,m)=2, 3$ or $4$ we construct new infinite family codes with at most five-weight.



\leftskip 0.0in
\rightskip 0.0in

\vskip 30pt

\section{Introduction}
\vskip 6pt
Let $\F_q$ be the finite field of order $q$, where $q=p^s$ is a power of a prime $p$ for some integer $s$. A $[n,k,d]$ linear code $C$ of length $n$ over $\F_q$ is a $k$-dimensional subspace of $\F_q^n$ with minimum Hamming distance $d$.  Let $A_i$ denote the number of codewords with Hamming weight $i$ in a code $C$ of length $n$. The weight enumerator of a linear code $C$ is defined to be the polynomial $A(z)=1 + A_1z + A_2z^2 +...+ A_nz^n$. The sequence $(1,A_1,...,A_n)$ is called the weight distribution of the code, and it is an important research topic to determine the weight distribution of a linear code in coding theory, since the weight distributions of codes contain crucial information, such  as to estimate the error correcting capability and the probability of error detection and correction of codes with respect to some decoding algorithms. A code $C$ is said to be a $t$-weight code if the number of nonzero $A_i$ in the sequence $(A_1, A_2,...,A_n)$ is equal to $t$. The class of codes with few weights are an important class in coding theory, and they have been studied in many articles, see for instance \cite{ding15,ding16,yu}.

\vskip 6pt
Let $r=q^m$ for some positive integer $m$, and let $\Tr(\cdot)$ denote the trace function from $\F_r$ to $\F_q$. Let  $D =\left\{d_1, d_2,...,d_n\right\}\subseteq \F^*_r$ be a subset of $\F^*_r$, where $\F_r^*=\F_r-\{0\}$, we define a linear code of length $n$ over $\F_q$ by
$$C_D = \left\{(\Tr(xd_1),\Tr(xd_2),...,\Tr(xd_n))\, | \, x \in \F_r \right\}.$$

The code $C_D$ is called a {\it trace code} over $\F_q$, and the set $D$ is called {\it defining set} of the code. It is well-known that many linear codes could be produced by selecting the defining set (see for instance \cite{ding15,ding08,ding14,Heng}). In a series of papers \cite{shi2016,shi2016_2}, the notion of trace codes has been extended from finite fields to finite rings as follows. Let $R$ be a finite commutative ring, and $R_m$ an extension of $R$ of degree $m$, a trace code with a defining set $L = \{d_1, d_2,...,d_n\} \subseteq R_m^*$ is defined by the following formula
\begin{equation}\label{tr_code}
C_L = \left\{(Tr(xd_1),Tr(xd_2),...,Tr(xd_n))\, |\, x \in R_m \right\},
\end{equation}
where $R_m^*$ is the multiplicative group of units of $R_m$, and $Tr(\cdot)$ is a linear function from $R_m$ to $R$.

\vskip 6pt
The class of finite rings of the form $R=\F_{q}+u\F_{q}, u^2=0$, has been used widely as alphabets in certain codes. Throughout this paper we let $R=\F_{q}+u\F_{q}$ and $\R=\F_{r}+u\F_{r}$ be an extension of $R$. The Lee weight distribution of the trace code $C_L$ define by Equation~($\ref{tr_code}$) is settled for a few special cases and is quite complex in general. The following special cases of trace codes have been studied in the literature.

\begin{itemize}
\item {(i)}  $R=\F_2+u\F_2, u^2=0$;  $\R=\F_{2^m}+u\F_{2^m}$, $L=\R^*=\F_{2^m}^*+u\F_{2^m}$. Then the code $C_L$ is a two-weight code respect to Lee weight (see \cite{shi2016}).
\item {(ii)}  $R=\F_p+u\F_p, u^2=0$;   $\R=\F_{p^m}+u\F_{p^m}$, $L={\cal Q}+ u\F_{p^m}\subsetneq \R^*$, where $p$ is an odd prime and ${\cal Q}$ is the set of all square elements of $\F_{p^m}^*$, then $C_L$ is a two-weight or three-weight code respect to Lee weight (see \cite{shi2016_2}).
\end{itemize}
In this paper, we let $R=\F_{q}+u\F_{q}, u^2=0$, where $q=p^s$, $p$ is a prime.  Let $\R=\F_{r}+u\F_{r}$,   and  $L=C_0^{(e,r)}+u\F_r\subseteq \R^*$, where $e$ is a divisor of $q-1$ and $C_0^{(e,r)}$ is the cyclotomic class of order $e$. The purpose of this paper is to investigate the Lee weight distribution of the trace code $C_L$ with defining set $L$.  We show that our main results include the results of \cite{shi2016} as a special case when $p=2$ and $s=1$. When $e=2$ and $q$ is odd, then $C_0^{(2,r)}={\cal Q}$. Therefore, our results also include the results of \cite{shi2016_2} as special cases. Furthermore, by taking $\gcd(e,m)=3$ or $4$ we construct new infinite family codes with at most five-weight.

\vskip 6pt
This paper is organized as follows. In Section 2, we introduce some basic notations and definitions which will be used later. In this section,  we also construct the code $C_q(m,e)$ defined in Equation~($\ref{ourcode}$) (in Section 3) for any positive divisor $e$ of $q-1$. Section 3 shows that the Lee weight distribution of the codes $C_q(m,e)$ is related to the value of $\gcd(e,m)$. When $\gcd(e,m)=1$, we show that the codes defined and investigated in $\cite{shi2016}$ are just special cases for $q=2$ (see Corollary $\ref{corollary11}$). Moreover, if $\gcd(e,m)=1$, our constructed codes achieve the Griesmer bound, and these codes are optimal for given length and dimension in few cases. We also investigate the cases where $\gcd(e,m)=2$ and $\gcd(e,m)=1$ and show that this cases are a generalization of the codes investigated in $\cite{shi2016_2}$. Furthermore, if $\gcd(e,m)=3$ or $\gcd(e,m)=4$, we construct new infinite families of codes with at most five-weight.

\section{Preliminaries}
Let $R=\F_{q}+u\F_{q}$, where $u^2=0$. A linear code $C$ of length $n$  over $R$ is an $R$-submodule of $R^n$. Let ${\bf x}=(x_1, x_2,...,x_n)$ and ${\bf y}=(y_1,y_2,...,y_n)$ be two vectors of $R^n$, the inner product of ${\bf x}$ and ${\bf y}$ is defined by $\langle {\bf x},{\bf y} \rangle = \sum_{i=1}^{n}x_iy_i$, where the operation is performed in $R$. The dual code $C^{\perp}$ of $C$ is defined as
$$C^\perp = \left\{ {\bf y}\in R^n \,|\, \langle {\bf x}, {\bf y}\rangle =0, \forall\,  {\bf x} \in C \right\}.$$ It is easy to verify that $C^\perp$ is also a linear code over $R$.

For any $a+ub\in R$, where $a,b \in \F_q$, the {\it Gray map} $\phi$ from $R$ to $\F_q^2$ is defined by
$$\phi: R\to \F_q^2, a+ub\mapsto (b,a+b).$$
Let ${\bf x}\in R^n$, then ${\bf x}$ can be written as ${\bf x}={\bf a}+u{\bf b}$, where ${\bf a}, {\bf b}\in \F_q^n$. The map $\phi$ is a bijection, which can be extended naturally from $R^n$ to $\F_q^{2n}$ as follows.
$$
\phi: R^n\to \F_q^{2n}, {\bf x}={\bf a}+u{\bf b}\mapsto ({\bf b},{\bf a}+{\bf b}).$$
The {\it Lee weight} of a vector ${\bf a}+u{\bf b}$ of length $n$ over $R$ is defined to be the Hamming weight of its Gray image as follows:
$$w_L({\bf a}+u{\bf b})=w_H({\bf b})+w_H({\bf a}+{\bf b}),$$
where ${\bf a}, {\bf b} \in \F_q^n$. The {\it Lee distance} of ${\bf x}, {\bf y} \in R^n$ is defined as $w_L({\bf x}- {\bf y})$. Thus the Gray map by construction is a linear isometry from $(R^n, d_L)$ to $(\F_q^{2n}, d_H)$.


Let $\R = \F_{r} + u\F_{r}, u^2=0$, where $r=q^m$ for some integer $m$,  be the extension of the ring $R$. It is easy to see that the ring $\R$ is a local ring with maximal ideal $M =\langle u\rangle$. The residue field $\R/ M$ is isomorphic to $\F_{r}$. The multiplicative group $\R^*$  of all units in $R$, is isomorphic to the product of a cyclic group of order $r-1$ and an elementary abelian group of order $r$.
\begin{definition}
{\it Let $F$ be the Frobenius operator over $\R$ defined by $F(a + ub)=a^q + ub^q$. The trace function, denoted by $Tr$ is then defined by
$$ Tr=\sum_{j=0}^{m-1}F^j: \R\to R, a+ub\mapsto \sum_{j=0}^{m-1}F^j(a+ub)=\sum_{j=0}^{m-1}(a^{q^j}+ub^{q^j}).$$}
\end{definition}
By the definition, it is easy to check that for any $a+ub\in \R$, where $a, b \in \F_{r}$, we have $Tr(a+ub)\in R$, and in particular, we have $$ Tr(a + ub) = \Tr(a) + u\Tr(b),$$
where $\Tr(\cdot)$ denotes the trace function from $\F_{r}$ to $\F_q$.


\begin{definition} {\it An additive character of $\F_{q}$ is a group homomorphism $\chi$ from $(\F_{q}, +)$ to the multiplicative group $(\mathbb{C}^*, \cdot)$, where $\mathbb{C}^*$ is the set of all nonzero complex numbers.}
\end{definition}
Let $Tr_{q/p}(\cdot)$ denote the trace function from $\F_{q}$ to $\F_p$, and let $i=\sqrt{-1}$ be the imaginary unit. It is easy to verify that for each $b\in \F_{q}$ , the function
$$ \chi_b(c)=e^{2i\pi Tr_{q/p}(bc)/p}, \quad \quad \text{ for all }c\in \F_{q},$$
defines an additive character of $\F_{q}$. When $b=0$, $\chi_0(c)=1 \text{ for all }c\in \F_{q}$, and it is called the {\it trivial additive character} of $\F_{q}$. The character $\chi_1(c)$ is called the {\it canonical additive character} of $\F_{q}$.

\begin{definition} {\it A multiplicative character of $\F_{q}$ is a group homomorphism  $\psi$ from $(\F_{q}^*, \cdot)$ to the multiplicative group $(\mathbb{C}^*, \cdot)$, where $\mathbb{C}^*$ is the set of all nonzero complex numbers.}
\end{definition}
Let $g$ be a fixed primitive element of $\F_{q}^*$. For each $j=0,1,...,q-2,$ the function $\psi_j$ with
\begin{equation}\label{multcar}
\psi_j(g^k)=e^{2i\pi jk/(q-1) } \quad\text{  for  }k=1,...,q-1,
\end{equation}
is a multiplicative character of $\F_{q}$. In particular, if  $j=0$, then $\psi_0(x)=1$ for all $x\in \F_{q}^*$, and this character is called the {\it trivial multiplicative character} of $\F_{q}$.

Let $q$ be odd and $j=(q-1)/2$ in Equation~(\ref{multcar}), then we get a multiplicative character such that
$$\eta(g^k)=\psi_{\frac{q-1}{2}}(g^k)=\left\{
\begin{array}{cc}
         1, & \text{ if }k \text{ is even,} \\
        -1, & \text{otherwise.}
    \end{array}
\right. $$
This multiplicative character $\eta$ is called the {\it quadratic character} of $\F_{q}$.

\begin{definition} {\it Let $\psi$ be a multiplicative character and $\chi$ an additive character of $\F_{q}$ respectively. The Gaussian sum $G(\psi,\chi)$ is defined as follows:
$$G(\psi,\chi)=\sum_{c\in \F_{q}^*}\psi(c)\chi(c).$$}
\end{definition}
\begin{remark}
In general, it is difficult to compute the value of the Gaussian sum $G(\psi, \chi)$. If $\psi \not= \psi_0$ and $\chi \not = \chi_0$, then the absolute value of $G(\psi,\chi)$ is $q^{1/2}$ (\cite{Lidl97}). If $q=p^s$, where $p$ is an odd prime and $s$ is a positive integer, then
\begin{equation}\label{gauss2}
G(\eta,\chi_1)=\left\{
\begin{array}{cc}
         (-1)^{s-1}q^{1/2}, & \text{ if } p\equiv 1 \pmod 4, \\
        (-1)^{s-1}i^sq^{1/2}, & \text{if }  p\equiv 3 \pmod 4.
    \end{array}
\right.
\end{equation}
\end{remark}
The following result (see \cite{Lidl97}) gives a relation between the Gaussian sum and the additive character and the multiplicative character of $\F_q$, which is also useful in the sequel.

\begin{lemma}\label{lemma1}
Let $\chi$ be a nontrivial additive character of $\F_{q}$ with $q$ being odd, and let $f(x) = a_2x^2+a_1x+a_0 \in \F_{q}[x]$ with $a_2\not= 0$. Then

$$\sum_{c\in \F_{q}}\chi(f(c)) = \chi(a_0-a_1^2(4a_2)^{-1})\eta(a_2)G(\eta,\chi).$$
\end{lemma}

In the following of this section, we shall introduce the notions of cyclotomy and Gaussian periods, and also give some known results on  periods polynomials.

Let $r-1 = nN$, where $n >1$ and $N > 1$, and suppose $\F_r^*=\langle \alpha\rangle$. The cosets $C^{(N,r)}_i =\alpha^i\langle\alpha^N\rangle$ for $i = 0, 1, . . . , N-1$ are called the {\it cyclotomic classes} of order $N$ in $\F_r$. In particular, if $i=0$, then $C^{(N,r)}_0=\langle \alpha^N\rangle$ is just the subgroup of $\F_r^*$ generated by $\alpha^N$.

\begin{lemma}[\cite{ma11}]\label{lemma 2.2}
Let $e$ be a positive divisor of $q-1$. Then we have the following multiset equality:

$$ \left\{xy \,|\, y \in C^{(e,r)}_0 , x \in \F^*_q  \right\}=\frac{q-1}{e}\gcd(m,e)\ast C^{(\gcd(m,e),r)}_0,$$
where $\frac{q-1}{e}\gcd(m,e)\ast C^{(\gcd(m,e),r)}_0$ denotes the multiset in which each element in the set $ C^{(\gcd(m,e),r)}_0$ appears in the multiset with multiplicity $\frac{q-1}{e}\gcd(m,e)$.
\end{lemma}

\vskip 6pt
\begin{definition} {\it
(i)\,\, Let $\kappa$ be the canonical additive character of $\F_{r}$. The  Gaussian periods  $\eta^{(N,r)}_i$, where $i=0, 1, \cdots, N-1$, are defined to be
$$ \eta^{(N,r)}_i= \sum_{x\in C_i^{(N,r)}}\kappa(x).$$

(ii)\,\, The  period polynomials $\psi_{(N,r)}(X)$ are defined by
$$\psi_{(N,r)}(X) =\prod_{i=0}^{N-1}\left( X-\eta_i^{(N,r)}\right).$$}
\end{definition}
In general, the values of the Gaussian periods are very hard to determine. However, they can be computed in a few cases. The following lemma is well-known, and also can be obtained from Lemma $\ref{lemma1}$ and Equation ($\ref{gauss2}$).

\begin{lemma}\label{lemma7}
If $N = 2$, then the Gaussian periods are given by the following
$$\eta_0^{(2,r)}=\left\{
		\begin{array}{cc}
			\frac{-1+(-1)^{sm-1}r^{1/2}}{2}, & \text{ if } p\equiv 1\pmod 4,\\
			\frac{-1+(-1)^{sm-1}i^{sm}r^{1/2}}{2}, & \text{ if } p\equiv 3\pmod 4,		
		\end{array}
\right.
$$
and $\eta_1^{(2,r)}=-1-\eta_0^{(2,r)}$.
\end{lemma}

It is known that $\psi^{(N,r)}(X)$ is a polynomial with integer coefficients (see \cite{Myerson}). When $N$ is small, for example, $N=3, 4$, the following four lemmas (see \cite{Myerson}) give a precise characterizations for these special  periods polynomials which will be used in the next section.

\begin{lemma}
Let $N = 3$. Let $c$ and $d$ be defined by $4r = c^2 + 27d^2$, $c \equiv 1 \pmod 3$, and, if $p \equiv 1 \pmod 3$, then $\gcd(c, p) = 1$.
These restrictions determine $c$ uniquely, and $d$ up to sign. Then we have
$$ \psi_{(3,r)}(X)=X^3+X^2-\frac{r-1}{3}X-\frac{(c+3)r-1}{27}.$$
\end{lemma}

Recall that $q=p^s, r=q^m$, and $r-1=nN$. For the decomposition of the period polynomial $\psi_{(3,r)}(X)$, we have
\begin{lemma}\label{Lemma 2.4}
Let $N = 3$. We have the following results on the factorization of $\psi_{(3,r)}(X)$. Then
\begin{description}
\item[(a)] If $p \equiv 2 \pmod 3$, then $sm$ is even, and
$$
\psi_{(3,r)}(X)=\left\{
\begin{array}{cc}
         3^{-3}(3X+1+2\sqrt{r})(3X+1-\sqrt{r})^2, & \text{ if } sm/2 \text{ even,} \\
        3^{-3}(3X+1-2\sqrt{r})(3X+1+\sqrt{r})^2, & \text{ if } sm/2 \text{ odd.}
    \end{array}
\right.$$
\item[(b)] If $p \equiv 1 \pmod 3$, and $sm\not\equiv 0 \pmod 3$, then $\psi_{(3,r)}(X)$ is irreducible over the rationals.
\item[(c)] If $p \equiv 1 \pmod 3$, and $sm \equiv 0 \pmod 3$, then
$$\psi_{(3,r)}(X) =\frac{1}{27}(3X + 1 - c_1r^{1/3})(3X + 1 + \frac{1}{2}(c_1 + 9d_1)r^{1/3})\cdot(3X + 1 + \frac{1}{2}(c_1-9d_1)r^{1/3}),$$
\end{description}
where $c_1$ and $d_1$ are given by $4p^{m/3} = c_1^2 + 27d_1^2, c_1 \equiv 1 \pmod 3$ and $\gcd(c_1, p) = 1$.
\end{lemma}

\begin{lemma}\label{lemma5}
Let $N = 4$. Let $u$ and $v$ be defined by $r = u^2 + 4v^2$, $u \equiv 1 \pmod 4$, and, if $p \equiv 1 \pmod 4$, then $\gcd(u, p) = 1$. These restrictions determine $u$ uniquely, and $v$ up to sign.

If $n$ is even, then
$$\psi_{(4,r)}(X) = X^4 + X^3 -\frac{3r-3}{8}X^2 + \frac{(2u-3)r+1}{16}X+\frac{r^2-(4u^2-8u+6)r+1}{256}.$$

If $n$ is odd, then
$$\psi_{(4,r)}(X) = X^4 + X^3 +\frac{r+3}{8}X^2 + \frac{(2u+1)r+1}{16}X+\frac{9r^2-(4u^2-8u-2)r+1}{256}.$$
\end{lemma}

\begin{lemma}\label{lemma6}
Let $N = 4$. We have the following results on the factorization of $\psi_{(4,r)}(X)$.
\begin{description}
\item[(a)] If $p \equiv 3 \pmod 4$, then $sm$ is even, and
$$\psi_{(4,r)}(X) =\left\{
		\begin{array}{cc}
			4^{-4}(4X+1+3\sqrt{r})(4X+1-\sqrt{r})^3, & \text{ if }sm/2\text{ even,}\\
			4^{-4}(4X+1-3\sqrt{r})(4X+1+\sqrt{r})^3, & \text{ if }sm/2\text{ odd.}\\
		\end{array}
		\right.
$$
\item[(b)] If $p \equiv 1 \pmod 4$, and $sm$ is odd, then $\psi_{(4,r)}(X)$ is irreducible over the rationals.
\item[(c)] If $p \equiv 1 \pmod 4$, and $sm \equiv 2 \pmod 4$, then
$$\psi_{(4,r)}(X) =4^{-4}((4X+1)^2+2\sqrt{r}(4X+1)-r-2u\sqrt{r})((4X+1)^2-2\sqrt{r}(4X+1)-r+2u\sqrt{r}),$$
the quadratics being irreducible, the $u$ is defined in Lemma \ref{lemma5}.
\item[(d)] If $p \equiv 1 \pmod 4$, and $sm \equiv 0 \pmod 4$, then
\begin{align*}
\psi_{(4,r)}(X)&=4^{-4}((4X+1)+\sqrt{r}+2r^{1/4}u_1)((4X+1)+\sqrt{r}-2r^{1/4}u_1)\\
	&\times((4X+1)-\sqrt{r}+4r^{1/4}v_1)((4X+1)-\sqrt{r}-4r^{1/4}v_1),
\end{align*}
\end{description}
where $u_1$ and $v_1$ are given by $p^{m/2} = u_1^2+ 4v_1^2, u_1 \equiv 1 \pmod 4$ and $\gcd(u_1, p) = 1$.
\end{lemma}



\section{The Lee weight distributions of trace codes over $R$}
Recall $\R=\F_r+u\F_r$, where $r=q^m$ for some integer $m$, and $R=\F_q+u\F_q$. Let $e$ be a positive divisor of $q-1$, let $L=C_0^{(e,r)}+u\F_r$, then $L$ is a subgroup of $\R^*$ of index $e$. For any $a\in \R$, we define the evaluation map from $\R$ to $R$ as follows.
$$ev: \R\to R, a\mapsto ev(a)=(Tr(ax))_{x\in L}. $$
 Accordingly,  A {\it trace code} $C_q(m,e)$  of length $n=\vert L \vert=\frac{r^2-r}{e}$ over $R$ is defined by
\begin{equation}\label{ourcode}
C_q(m,e)=\{ev(a)\, | \, a \in \R\}.
\end{equation}

In this section, we shall study the Lee weight distribution of the code $C_q(m,e)$.

Let $\chi$ and $\kappa$ be the canonical characters of $\F_q$ and $\F_r$ respectively. For any $y\in \F_q$, the following is a well-known property of additive characters of $\F_q$.
\begin{equation}\label{car}
\sum\limits_{x\in\F_q}\chi(xy)=\left\{\begin{array}{ll}
				q, &\text{if }y=0,\\
				0, &\text{if }y\not=0.
\end{array} \right.
\end{equation}
Let $N$ be a positive integer. For any ${\bf y} = (y_1,y_2,...,y_N)\in \F^N_q$ , let
$$\Theta({\bf y})=\sum\limits_{j=1}^{N}\chi(y_j).$$
For simplicity, we let $\theta(a) = \Theta(\phi(ev(a)))$, where $a \in \R$, and $\phi$ is the Gray map from $R$ to $\F_q^2$. By linearity of the Gray map, and the evaluation map, we see that $\theta(sa) = \Theta(\phi(ev(sa)))$, for any $s\in \F_q$.

\begin{proposition}\label{proposition8}
Let $N$ be a positive integer and let ${\bf y} = (y_1,y_2,...,y_N) \in \F^N_q$. Then

$$\sum\limits_{s \in \F_q^*} \Theta(s{\bf y}) =(q-1)N-qw_H({\bf y}).$$
\end{proposition}
\begin{proof}
By the definition of $\Theta$ and Equation (\ref{car}), we get
$$\sum\limits_{s \in \F_q} \Theta(s{\bf y}) = \sum\limits_{s \in \F_q} \sum\limits_{j=1}^{N} \chi(sy_j)=q \sum\limits_{j=1}^{N}\frac{1}{q}\sum\limits_{s \in \F_q} \chi(sy_j)=q(N-w_H({\bf y}))=qN-qw_H({\bf y}).$$
Note that
$$\Theta({\bf 0}) = \sum\limits_{j=1}^{N} \chi(0)=N.$$
Hence, the result follows.
\end{proof}

\begin{proposition}\label{proposition9}
Let $a\in \R$ such that $a=u\beta \in M \setminus \left \{ 0 \right\}$ where $M=\langle u\rangle$, then we have
$$\sum\limits_{s \in \F_q^*} \theta(sa) =2r \frac{q-1}{e}\gcd(e,m)\sum\limits_{s \in C^{\left(\gcd(e,m),r \right)}_0}\kappa(\beta s).$$
\end{proposition}
\begin{proof}
We have that
\begin{align*}
\sum\limits_{s \in \F_q^*} \theta(sa) =\sum\limits_{s \in \F_q^*}\Theta \left( \phi \left( ev(as) \right) \right)=\sum\limits_{s \in \F_q^*}\Theta \left( \phi \left( Tr(asx) \right)_{x \in L} \right).
\end{align*}
Note that $x=t+ut^\prime$, where $t \in C_0^{(e,r)}$ and $t^\prime\in \F_q$. Therefore, $ax=u\beta t$ and
$$Tr(ax)=Tr(u\beta t)=u\Tr(\beta t).$$
Let $s \in \F_q^*$, taking Gray map yields
$$\phi(ev(as))=\phi((Tr(asx))_{x\in L})=(\Tr(\beta st),\Tr(\beta st))_{t,t^\prime}.$$
This implies that
\begin{align*}
\sum\limits_{s \in \F_q^*} \theta(sa) &= \sum\limits_{s \in \F_q^*}\Theta(\phi(ev(as)))=\sum\limits_{s \in \F_q^*}\Theta((\Tr(\beta st),\Tr(\beta st))_{t,t^\prime})\\
   &= 2\sum\limits_{s \in \F_q^*} \sum\limits_{t \in C_0^{(e,r)}}\sum\limits_{t^\prime \in \F_r}\chi(\Tr(\beta st))\\
   &=2r\sum\limits_{s \in \F_q^*}\sum\limits_{t \in C^{\left(e,r \right)}_0}\kappa(\beta st).
\end{align*}
Using Lemma \ref{lemma 2.2} we get
$$\sum\limits_{s \in \F_q^*} \theta(sa) =2r \frac{q-1}{e}gcd(m,e) \sum\limits_{t \in C^{\left(\gcd(e,m),r \right)}_0}\kappa(\beta t).\qedhere$$
\end{proof}

In the following, we will determine the Lee weight distribution of the code $C_q(m,e)$ when $\gcd(m,e)=1,2,3,4$. In particular, we will show that there exist optimal codes over the finite field $\F_q$ which are the Gray image of  these class of trace codes.
\begin{theorem}\label{theorem10}
Let $a\in \R$ and $\gcd(e,m)=1$, then $C_q(m,e)$ is a code of length ${{(r^2-r)}\over e}$, size $r^2$, and its Lee weight enumerator is determined by
$$
A_{C_q(m,e)}(z)=1+(r^2-r)z^{2\frac{q-1}{eq}(r^{2}-r)}+(r-1)z^{2\frac{q-1}{eq} r^{2}}.
$$


In particular, the code $C_q(m,e)$ is a two-weight code.
\end{theorem}

\begin{proof}
If $a=0$, then $Tr(ax)=0$. So $w_L(ev(a)) = 0$.

If $a=u\beta$ with $\beta \in \F_r^*$, using Proposition \ref{proposition9} we get
\begin{align*}
\sum\limits_{s \in \F_q^*} \theta(sa) &=2r\frac{q-1}{e}\sum\limits_{s \in C^{\left(1,r \right)}_0}\kappa(\beta s)
=2r\frac{q-1}{e}\sum\limits_{s \in \F_{r}^*}\kappa(\beta s)=-2r\frac{q-1}{e}.
\end{align*}
Note that $ev(a) \in R^n$, hence $\phi(ev(a)) \in \F^{2n}_{q}$. By Proposition \ref{proposition8} we get
$$ -2r\frac{q-1}{e}=2(q-1)N-qw_H(\phi(ev(a))),$$
which implies that
$$w_L(ev(a))=2\frac{q-1}{eq}r^{2}.$$
If $a \in \R^*$ and $x=t+ut^\prime\in L$, then $a$ can be expressed as $a=\alpha+u\beta$. Therefore, $ax=\alpha t+u(\alpha t^\prime+\beta t)$, and
$$ Tr(ax) = \Tr(\alpha t) + u\Tr(\alpha t^\prime + \beta t), $$
which gives that
$$\phi(ev(a))=\phi((Tr(ax))_{x\in L})=(\Tr(\alpha t^\prime + \beta t),\Tr(\alpha t^\prime + \beta t)+\Tr(\alpha t))_{t,t^\prime}.$$
Taking character sum yields
\begin{align*}
\theta(a) &=\sum\limits_{t \in C_0^{(e,r)}} \chi(\Tr(\beta t))\sum\limits_{t^\prime \in \F_{r}} \chi(\Tr(\alpha t^\prime))+\sum\limits_{t \in C_0^{(e,r)}}\chi(\Tr((\alpha+\beta) t))
\sum\limits_{t^\prime \in \F_{r}} \chi(\Tr(\alpha t^\prime))
   \\&=\sum\limits_{t \in C_0^{(e,r)}} \kappa(\beta t)\sum\limits_{t^\prime \in \F_{r}} \kappa(\alpha t^\prime)+\sum\limits_{t \in C_0^{(e,r)}}\kappa((\alpha+\beta) t)
\sum\limits_{t^\prime \in \F_{r}} \kappa(\alpha t^\prime).
\end{align*}
Since $\sum\limits_{t^\prime \in \F_{r}} \kappa(\alpha t^\prime)=0$, then $\theta(a)=0$. Using Proposition \ref{proposition8} we get
$$ w_L(ev(a))=2\frac{q-1}{q}n=2\frac{q-1}{eq}(r^{2}-r).\qedhere$$
\end{proof}
\begin{corollary}\label{corollary11}
Take $q=2$, then the code $C_2(m,e)$ is a two-weight code with the weights $2^{2m}$ and $2^{2m}-2^m$, of frequencies $2^m-1$ and $2^{2m}-2^m$ respectively.
\end{corollary}
\begin{proof}
If $q=2$ then $e=1$, by Theorem \ref{theorem10} and the sizes of $M$ and $\R^*$, the proof can be easily achieved.
\end{proof}

Let $C$ be an $[n, k, d]$ code over $\F_q$ with $k \geq 1$. The following bound is called  the Griesmer bound:
$$\sum_{j=0}^{k-1}\left\lceil \frac{d}{q^j} \right\rceil \leq n.$$
We show that, in the following, the image of the code $C_q(m,e)$ under Gray map can reach the Griesmer bound in some cases. That means we can obtain some optimal codes from trace codes.
\begin{theorem}\label{theorem12}
Assume $\gcd(m,e)=1$, and let $C_q(m,e)$ be the code over $R$. Then we have

\begin{itemize}
\item If $e = 1$, then for the code $\phi(C_q(m,e))$, we have $\sum_{j=0}^{k-1}\left\lceil \frac{d}{q^j} \right\rceil = n-1.$

\item If $e\geq 2$, then the code $\phi(C_q(m,e))$ meets the Griesmer bound with equality.
\end{itemize}

\end{theorem}

\begin{proof}

If $\gcd(m,e)=1$, the code $\phi(C_q(m,e))$ have the parameters $$n=\frac{2}{e}(q^{2m}-q^m),\quad k=2m,\quad d=\frac{2(q-1)}{eq}(r^2-r)=\frac{2(q-1)}{eq}(q^{2m}-q^m).$$

If $e\geq 2$ the ceiling function takes two values depending on $j$.

\begin{itemize}
\item $0 \leq j \leq m-1$. In this case, $\left\lceil \frac{d}{q^j}\right\rceil=\frac{2(q-1)}{e}(q^{2m-j-1}-q^{m-j-1})$,

\item $m \leq j \leq 2m-1$. In this case, $\left\lceil \frac{d}{q^j}\right\rceil=\frac{2(q-1)q^{2m-j-1}}{e}$.
\end{itemize}

Therefore,

\begin{align*}
\sum_{j=0}^{2m-1}\left\lceil \frac{d}{q^j} \right\rceil&=\sum_{j=0}^{m-1}\frac{2(q-1)}{e}(q^{2m-j-1}-q^{m-j-1})+\sum_{j=m}^{2m-1}\frac{2(q-1)q^{2m-j-1}}{e}\\
   &=\frac{2(q-1)}{e}\left(  \sum_{j=0}^{2m-1}q^{2m-j-1}-\sum_{j=0}^{m-1}q^{m-j-1} \right)\\
   &=\frac{2}{e}(q^{2m}-q^m)=n.
\end{align*}

If $e=1$ the ceiling function takes three values depending on $j$.

\begin{itemize}
\item $0 \leq j \leq m-1$. In this case, $\left\lceil \frac{d}{q^j}\right\rceil=\frac{2(q-1)}{e}(q^{2m-j-1}-q^{m-j-1})$,

\item  $j=m$. In this case, $\left\lceil \frac{d}{q^j}\right\rceil=\frac{2(q-1)q^{m-1}}{e}-1$,

\item $m+1 \leq j \leq 2m-1$. In this case, $\left\lceil \frac{d}{q^j}\right\rceil=\frac{2(q-1)q^{2m-j-1}}{e}$.
\end{itemize}

Similarly as the case $e\geq 2$, the result can be obtained for $e=1$.
\end{proof}

\begin{remark}
Theorem 4.3 of \cite{shi2016} and Theorem 2 of \cite{shi2016_2} are just a special cases of Theorem \ref{theorem10}. When $\gcd(e,m)=1$, most of the codes $\phi(C_q(m,e))$ are optimal or nearly optimal for given length and dimension \cite{table}. We list some examples in Table \rom{4}.
\end{remark}
\begin{table}[h!]
\centering
\caption*{Table \rom{4}: Numeral examples of the code $C_q(m,e)$ in Theorem \ref{theorem10} }
 \begin{tabular}{||c c c c c c||}
 \hline
 $q$ & $m$ & $e$ & $[n,k,d]$ & Weight distribution & Optimality \\ [0.5ex]
 \hline\hline
 $2$ & $3$ & $1$ & $[112,6,56]$ & $1+56z^{56}+7z^{64}$ & Optimal \\
 $3$ & $2$ & $1$ & $[144,4,96]$ & $1+72z^{96}+8z^{108}$ & Optimal \\
 $4$ & $2$ & $3$ & $[160,4,120]$ & $1+240z^{120}+15z^{128}$ & Optimal \\
  $5$ & $1$ & $1$ & $[40,2,32]$ & $1+20z^{32}+4z^{40}$ & Nearly Optimal\\
 $7$  & $1$ & $2$ & $[42,2,36]$ & $1+42z^{36}+6z^{42}$ & Optimal \\
 $9$ & $1$ & $2$ & $[72,2,64]$ & $1+72z^{64}+ 8z^{72}$ & Optimal \\
 \hline
 \end{tabular}
\end{table}


\begin{theorem}\label{theorem13}
Let $e$ be a positive divisor of $q-1$. If $\gcd(e,m)=2$ then the code $C_q(m,e)$ is a three-weight code of length ${{(r^2-r)}\over e}$, size $r^2$, and its Lee weight enumerator is determined by

$$A_{C_q(m,e)}(z)=1+\frac{r-1}{2}z^{2\frac{q-1}{eq}(r^2-r^{3/2})} +(r^2-r)z^{ 2\frac{q-1}{eq}(r^2-r) }  +\frac{r-1}{2}z^{2\frac{q-1}{eq}(r^2+r^{3/2})} .$$

\end{theorem}
\begin{proof}
Since $\gcd(m, e) = 2$, then $m$ is even and $q$ is odd. If $a=0$, then $Tr(ax)=0$. So $w_L(ev(a)) = 0$.

If $a \in \R^*$ then $a$ can be expressed as $a=\alpha+u\beta$, let $x=t+ut^\prime\in L$. So $ax=\alpha t+u(\alpha t^\prime+\beta t)$, and
$$ Tr(ax) = \Tr(\alpha t) + u\Tr(\alpha t^\prime + \beta t), $$
which implies that
$$\phi(ev(a))=\phi((Tr(ax))_{x\in L})=(\Tr(\alpha t^\prime + \beta t),\Tr(\alpha t^\prime + \beta t)+\Tr(\alpha t))_{t,t^\prime}.$$
Taking character sum yields
\begin{align*}
\theta(a) &=\sum\limits_{t \in C_0^{(e,r)}} \chi(\Tr(\beta t))\sum\limits_{t^\prime \in \F_{r}} \chi(\Tr(\alpha t^\prime))+\sum\limits_{t \in C_0^{(e,r)}}\chi(\Tr((\alpha+\beta) t))
\sum\limits_{t^\prime \in \F_{r}} \chi(\Tr(\alpha t^\prime))
   \\&=\sum\limits_{t \in C_0^{(e,r)}} \kappa(\beta t)\sum\limits_{t^\prime \in \F_{r}} \kappa(\alpha t^\prime)+\sum\limits_{t \in C_0^{(e,r)}}\kappa((\alpha+\beta) t)
\sum\limits_{t^\prime \in \F_{r}} \kappa(\alpha t^\prime).
\end{align*}
Since $\sum\limits_{t^\prime \in \F_{r}} \kappa(\alpha t^\prime)=0$, then $\theta(a)=0$. Using Proposition \ref{proposition8} we get

$$ w_L(ev(a))=2\frac{q-1}{q}n=2\frac{q-1}{eq}(r^2-r).$$
If $a \in M \setminus \left\{ 0 \right\}$, then $a=\beta u$ with $\beta \in \F_{r}^*$. By Proposition \ref{proposition9} we get
\begin{align*}
\sum\limits_{s \in \F_q^*} \theta(sa) &=4r\frac{q-1}{e}\sum\limits_{t \in C^{\left(2,p^m \right)}_0}\kappa(\beta t).
\end{align*}
If $\beta \in C_i^{(2,r)}$, then we have that
$$\sum\limits_{s \in \F_q^*} \theta(sa)=4r\frac{q-1}{e}\eta^{(2,r)}_i .$$
Note  that $ev(a) \in R^n$, then $\phi(ev(a)) \in \F^{2n}_{q}$, By Proposition \ref{proposition8} we get
$$ 4r\frac{q-1}{e}\eta^{(2,r)}_i=2(q-1)n-qw_H(\phi(ev(a))),$$
which implies that
$$w_L(ev(a))=2\frac{q-1}{eq} \left( en-2r\eta^{(2,r)}_i\right)=2\frac{q-1}{eq} \left( r^2-r(1+2\eta^{(2,r)}_i) \right) .$$
Note that $\vert C^{(2,r)}_0 \vert=\vert C^{(2,r)}_1 \vert=\frac{r-1}{2}$. The weight distribution of the code follows from the Lemma \ref{lemma7} and the discussion above. This completes the proof.
\end{proof}

\vskip 6pt
\begin{remark} If we take $e=2$, then Theorem 1 of \cite{shi2016_2} is just a special cases of our Theorem \ref{theorem13}.
\end{remark}
The following is  a numeral example.

\begin{example}
Let $q= 3$ and $m= 2$ and $e=2$. Then $\phi(C_3(2,2))$ is a three-weight ternary code of parameters $[72, 4, 36]$ with the weight enumerator

$$1+4z^{36}+72z^{48}+4z^{72}.$$
\end{example}

\begin{theorem}\label{theorem14}
Let $e$ be a divisor of $q-1$. When $\gcd(m,e) = 3$ the code $C_q(m,e)$ is with the Lee weight distribution given in Table \rom 2, where $c_1$ and $d_1$ are uniquely given by $4p^{m/3} = c_1^2 + 27d_1^2$, $c_1 \equiv 1 \pmod 3$ and $\gcd(c_1, p) = 1$.

\begin{center}
	\begin{table}
    \caption*{Table \rom{2}: Lee weight distribution of the code $C_q(m,e)$: $\gcd(m,e)=3$}
    \begin{minipage}{.5\linewidth}
      \caption*{ $p \equiv 1 \pmod 3$}
      \centering
        \begin{tabular}{|c|c|}
  \hline
	Lee weight &Frequency\\
  \hline
  \hline
	0&1\\
  \hline
	$2\frac{q-1}{eq}(r^2-r)$& $r^2-r$\\
  \hline
	$2\frac{q-1}{eq}(r^2-c_1r^{4/3})$& $\frac{r-1}{3}$\\
  \hline
	$2\frac{q-1}{eq}(r^2+\frac{1}{2}(c_1+9d_1)r^{4/3})$ &$\frac{r-1}{3}$\\
  \hline
  	$2\frac{q-1}{eq}(r^2+\frac{1}{2}(c_1-9d_1)r^{4/3})$ &$\frac{r-1}{3}$\\
  \hline
\end{tabular}
    \end{minipage}%
    \begin{minipage}{.5\linewidth}
      \centering
        \caption*{$p \equiv 2 \pmod 3$ and $sm\equiv 2 \pmod 4$}
        \begin{tabular}{|c|c|}
  \hline
	Lee weight &Frequency\\
  \hline
  \hline
	0&1\\
  \hline
	$2\frac{q-1}{eq}(r^2-r)$& $r^2-r$\\
  \hline
	$2\frac{q-1}{eq}(r^2+2r^{3/2})$& $\frac{r-1}{3}$\\
  \hline
	$2\frac{q-1}{eq}(r^2-r^{3/2})$ &$2\frac{r-1}{3}$\\
  \hline
\end{tabular}
    \end{minipage}
	\end{table}
\end{center}
\end{theorem}

\begin{proof}
If $a=0$, then $Tr(ax)=0$. So $w_L(ev(a)) = 0$.

If $a \in \R^*$ then $a$ can be expressed as $a=\alpha+u\beta$, where $\alpha\ne 0$, let $x=t+ut^\prime\in L$. Then $ax=\alpha t+u(\alpha t^\prime+\beta t)$, and
$$ Tr(ax) = \Tr(\alpha t) + u\Tr(\alpha t^\prime + \beta t), $$
which implies that
$$\phi(ev(a))=\phi((Tr(ax))_{x\in L})=(\Tr(\alpha t^\prime + \beta t),\Tr(\alpha t^\prime + \beta t)+\Tr(\alpha t))_{t,t^\prime}.$$
Taking character sum yields
\begin{align*}
\theta(a) &=\sum\limits_{t \in C_0^{(e,r)}} \chi(\Tr(\beta t))\sum\limits_{t^\prime \in \F_{r}} \chi(\Tr(\alpha t^\prime))+\sum\limits_{t \in C_0^{(e,r)}}\chi(\Tr((\alpha+\beta) t))
\sum\limits_{t^\prime \in \F_{r}} \chi(\Tr(\alpha t^\prime))
   \\&=\sum\limits_{t \in C_0^{(e,r)}} \kappa(\beta t)\sum\limits_{t^\prime \in \F_{r}} \kappa(\alpha t^\prime)+\sum\limits_{t \in C_0^{(e,r)}}\kappa((\alpha+\beta) t)
\sum\limits_{t^\prime \in \F_{r}} \kappa(\alpha t^\prime).
\end{align*}
Since $\sum\limits_{t^\prime \in \F_{r}} \kappa(\alpha t^\prime)=0$, then $\theta(a)=0$. Using Lemma \ref{proposition8} we get

$$ w_L(ev(a))=2\frac{q-1}{q}n=2\frac{q-1}{eq}(r^2-r).$$
If $a \in M \setminus \left\{ 0 \right\}$, then $a=\beta u$ with $\beta \in \F_{p^m}^*$. By the assumption $\gcd(m,e)=3$. It then follows from Proposition \ref{proposition9} that
\begin{align*}
\sum\limits_{s \in \F_q^*} \theta(sa) &=6r\frac{q-1}{e}\sum\limits_{t \in C^{\left(3,r\right)}_0}\kappa(\beta t).
\end{align*}
If $\beta \in C_i^{(3,r)}$
$$\sum\limits_{s \in \F_q^*} \theta(sa)=6r\frac{q-1}{e}\eta^{(3,r)}_i.$$
We have that $ev(a) \in R^n$, then $\phi(ev(a)) \in \F^{2n}_{q}$, By Proposition \ref{proposition8} we get
$$ 6r\frac{q-1}{e}\eta^{(3,r)}_i=2(q-1)n-qw_H(\phi(ev(a))),$$
which implies that
$$w_L(ev(a))=2\frac{q-1}{eq} \left( en-3r\eta^{(3,r)}_i\right)=2\frac{q-1}{eq} \left( r^2-r(1+3\eta^{(3,r)}_i) \right) .$$
Since $\gcd(m,e) = 3$, then $m \equiv 0 \pmod 3$. If $p \equiv 1 \pmod 3$ by Lemma $\ref{Lemma 2.4}$ the Gaussian periods $\eta^{(3,r)}_i$ take only the following three distinct values:
$$ \frac{-1+c_1r^{1/3}}{3}, \quad \frac{-1-\frac{1}{2}(c_1+9d_1)r^{1/3}}{3}, \quad \frac{-1-\frac{1}{2}(c_1-9d_1)r^{1/3}}{3} $$
If $p \equiv 2 \pmod 3$ the Gaussian periods $\eta^{(3,r)}_i$ take only two distinct values. Note that $\vert C^{(3,r)}_0 \vert=\vert C^{(3,r)}_1 \vert=\vert C^{(3,r)}_2 \vert=\frac{r-1}{3}$. The weight distribution of the code follows. This completes the proof.
\end{proof}

\begin{example}
Let $q= 7$ and $m= 3$ and $e=6$. Then $\phi(C_7(3,6))$ is a four-weight code of parameters $[39102, 6, 329]$ with the weight enumerator

$$1+117306z^{33516}+144z^{329}+144z^{37044}+144z^{30870}.$$
\end{example}

\begin{example}
Let $q= 4$ and $m= 3$ and $e=3$. Then $\phi(C_4(3,3))$ is a three-weight code of parameters $[2688, 6, 1536]$ over $\F_4$ with the weight enumerator

$$1+12z^{1536}+4032z^{2016}+24z^{2304}.$$
\end{example}
The last theorem of this section is the weight distribution of the code $C_q(m,e)$, where $\gcd(m,e)=4$.
\begin{theorem}
Let $e$ be a divisor of $q-1$. Suppose that $\gcd(m,e) = 4$, then the code $C_q(m,e)$ is with Lee weight distribution given in Table \rom 3.
\begin{center}
\begin{table}
\caption*{Table \rom{3}: Lee weight distribution of the code $C_q(m,e)$:  $\gcd(e,m)=4$}

\begin{minipage}{.5\linewidth}
\caption*{ $p \equiv 1 \pmod 4$}
      \centering
\begin{tabular}{|c|c|}

  \hline
	Lee weight &Frequency\\
  \hline
  \hline
	0&1\\
  \hline
	$2\frac{q-1}{eq}(r^2-r)$& $r^2-r$\\
  \hline
	$2\frac{q-1}{eq}(r^2-r^{3/2}-2u_1r^{5/4})$& $\frac{r-1}{4}$\\
  \hline
	$2\frac{q-1}{eq}(r^{2}-r^{3/2}+2u_1p^{5/4})$ &$\frac{r-1}{4}$\\
  \hline
  	$2\frac{q-1}{eq}(r^{2}+r^{3/2}-4v_1r^{5/4})$ &$\frac{r-1}{4}$\\
  \hline
    	$2\frac{q-1}{eq}(r^{2}+r^{3/2}+4v_1r^{5/4})$ &$\frac{r-1}{4}$\\
  \hline
\end{tabular}
\end{minipage}
\begin{minipage}{.5\linewidth}
\caption*{ $p \equiv 3 \pmod 4$}
      \centering
\begin{tabular}{|c|c|}

  \hline
	Lee weight &Frequency\\
  \hline
  \hline
	0&1\\
  \hline
	$2\frac{q-1}{eq}(r^2-r)$& $r^2-r$\\
  \hline
	$2\frac{q-1}{eq}(r^{2}-3r^{3/2})$& $\frac{r-1}{4}$\\
  \hline
	$2\frac{q-1}{eq}(r^{2}+r^{3/2})$ &$3\frac{r-1}{4}$\\
  \hline
\end{tabular}
\end{minipage}%

\end{table}

\end{center}
\end{theorem}
\begin{proof}
Because $\gcd(m,e) = 4$, then $m \equiv 0 \pmod 4$. If $p \equiv 1 \pmod 4$, then the Gaussian periods $\eta^{(4,p^m)}_i$ take only the following distinct values:
$$ \frac{(1+r^{1/2}+2u_1r^{1/4})}{4},\quad \frac{(1+r^{1/2}-2u_1r^{1/4})}{4},$$
and
$$\frac{(1-r^{1/2}+4v_1r^{1/4})}{4}, \quad \frac{(1-r^{1/2}-4v_1r^{1/4})}{4}.$$
Using a similar proof of Theorem \ref{theorem14}, the proof can be easily achieved.
\end{proof}

\begin{example}
Let $q= 5$ and $m= 4$ and $e=4$. Then $\phi(C_5(4,4))$ is a four-weight  code of parameters $[195000, 8, 142500]$ with the weight enumerator
$$1+390000z^{156000}+156z^{142500}+156z^{157500}+156z^{152500}+156z^{172500}.$$
\end{example}

\section{Conclusion}
 In this paper, we characterize the Lee weight distribution of the code $C_q(m,e)$ for all $e$ with $1\leq \gcd(e,m) \leq 4 $. The period polynomial $\psi_{(N,r)}(X)$ and its factorization were determined for $N=5$ by Hoshi \cite{hoshi}, and for those dividing $8$ and $12$ by Gurak \cite{gurak}. This means that the Lee weight distribution of the code $C_q(m,e)$ for $\gcd(e,m)\in \{5,6,8,12\}$ also can be determined. However, the weight formulas will be messy. This demonstrates that the Lee weight distribution of the codes $C_q(m,e)$ is indeed very complicated.

 On the other hand, the codes constructed in this paper are a generalization of those studied in \cite{shi2016_2,shi2016}. In particular, If $q=2$ our result coincide with \cite{shi2016} and with \cite{shi2016_2} if $e=2$. Furthermore, our technique can be used to generalize other results such as in \cite{shiarx}.

\vskip 16pt
\noindent {\bf Acknowledgement.} This work was supported by NSFC under Grant No. 11171370.

\end{document}